\documentclass[letterpaper, 10 pt, conference]{ieeeconf}  

\IEEEoverridecommandlockouts                              
\usepackage{graphicx} 
\usepackage{amsmath} 
\usepackage{amssymb}  
\usepackage{color}
\usepackage{epstopdf}

\usepackage[normalem]{ulem}

\usepackage{amsthm} 
\newtheorem{theorem}{Theorem}[section]

\newtheorem{proposition}[theorem]{Proposition}
	
\newtheorem{remark}[theorem]{Remark}	
\newtheorem{problem}[theorem]{Problem}

\title{\LARGE \bf 
Mean-Field Controllability and Decentralized Stabilization of Markov Chains, Part II: Asymptotic Controllability and Polynomial Feedbacks}

\author{Shiba Biswal, Karthik Elamvazhuthi, and Spring Berman
\thanks{This work was supported by NSF Awards CMMI-1363499 and CMMI-1436960, by
ONR Young Investigator Award N00014-16-1-2605, and by the Arizona State University Global Security Initiative.}
\thanks{Shiba Biswal, Karthik Elamvazhuthi, and Spring Berman are with the
	School for Engineering of Matter, Transport and Energy, Arizona State University,
	Tempe, AZ, 85281 USA,
	{\tt\small \{sbiswal, karthikevaz, Spring.Berman\}@asu.edu}}%
\thanks{
	{\tt\small }}%
}

\begin{document}

\maketitle
\thispagestyle{empty}
\pagestyle{empty}

\begin{abstract}

This paper, the second of a two-part series, presents a method for mean-field feedback stabilization of a swarm of agents on a finite state space whose time evolution is modeled as a continuous time Markov chain (CTMC). The resulting (mean-field) control problem is that of controlling a nonlinear system with desired global stability properties. We first prove that any probability  distribution with a \textit{strongly connected support} can be stabilized using time-invariant inputs. Secondly, we show the asymptotic controllability of all possible probability distributions, including distributions that assign zero density to some states and which do not necessarily have a strongly connected support. Lastly, we demonstrate that there always exists a globally asymptotically stabilizing decentralized density feedback law with the additional property that the control inputs are zero at equilibrium, whenever the graph is strongly connected and bidirected. Then the problem of synthesizing closed-loop polynomial feedback is framed as a optimization problem using state-of-the-art sum-of-squares optimization tools. The optimization problem searches for polynomial feedback laws that make the candidate Lyapunov function a stability certificate for the resulting closed-loop system. Our methodology is tested for two cases on a five-vertex graph, and the stabilization properties of the constructed control laws are validated with numerical simulations of the corresponding system of ordinary differential equations.

\end{abstract}

\section{INTRODUCTION} \label{section:introduction}

This paper addresses the problem of redistributing a large number of homogeneous agents among a set of states, such as tasks to be performed or spatial locations to occupy. While there exist several well established methods for control of multi-agents \cite{ren2008distributed,bullo2009distributed,mesbahi2010graph}, many of these control approaches do not scale well to very large agent populations.  
Hence an alternative approach for controlling multi-agent systems is by modeling the system as a fluid. This is justified by modeling each agent's dynamics by a continuous-time Markov chain (CTMC) and then the mean-field behavior of the system is determined by the {\it Kolmogorov forward equation} corresponding to the CTMC.


A similar approach also exists when the agent dynamics evolves of discrete time. In this case the agents' state evolution over time is described by a discrete time Markov chain (DTMC). It is known that a discrete time Markov chain (DTMC) admits a stationary distribution under certain conditions of irreducibility, recurrency, and aperiodicity. These conditions are determined by the properties of the stochastic transition matrix of the process. If it is feasible to make a desired distribution invariant by choosing appropriate transition probabilities, then it is possible to compute {\it optimized} transition probabilities that guarantee the fastest rate of convergence to the invariant distribution. One of the first works to address this problem is \cite{boyd2004fastest}, which formulates a semidefinite program (SDP) whose solution is the set of transition probabilities that yield optimal exponential convergence.
In the case of a continuous time Markov chain (CTMC), the time evolution of the system is governed by a transition rate matrix, the \textit{generator} of the stochastic process. In \cite{berman2009optimized}, the authors present methods for computing optimized transition rates of a CTMC that drive the system to any strictly positive desired distribution at a fast convergence rate. 

The works \cite{berman2009optimized} and \cite{boyd2004fastest} address an open-loop optimal control problem for the Kolmogorov forward equation, in which the control parameters, which are the transition probabilities or rates of the process, are constrained to be time-invariant. These approaches have been extended to the case of time-varying control parameters in several different contexts. In \cite{bandyopadhyay2013inhomogeneous,demir2015decentralized,mather2011distributed}, the authors design feedback controllers to drive a Markov chain to a target distribution. In contrast to traditional control approaches for Markov chains that use only the agent states  as feedback  \cite{puterman2014markov}, these works use the agent {\it densities} at different states as feedback and continuously re-compute the control parameters such that the target distribution is stabilized. Since this type of feedback generally requires global information about the densities at all states, these works have developed  {\it decentralized} control approaches, in which each agent's control parameters depend only on information that the agent can obtain from its local environment. This information may be derived from activity at the agent's current state or from activity that is communicated from an adjacent state. Such approaches minimize the inter-agent communication that is needed to implement the control strategy.  There has also been some recent work on mean-field games, where Hamilton-Jacobi-Bellman (HJB) based methods are used for control sysnthesis \cite{gomes2013continuous}. However, unlike HJB based methods in classical control theory, mean-field games based approaches do not result in feedback controllers. In this framework the synthesized control inputs are open-loop in nature and have the desired behavior only for a predefined fixed initial conditions of the mean-field model.

In this paper, the second of a two-part series (Part I is \cite{elamvaz2016lin}), we contribute three main results to the mean-field control problem for CTMCs. 
First, we demonstrate that it is possible to compute density-independent transition rates of a CTMC that make any probability distribution with a strongly connected support (to be defined later) invariant and globally stable. Similar work in \cite{acikmese2012markov} has characterized the class of stabilizable stationary distributions for DTMCs with control parameters that are time and density-invariant; we characterize this class of distributions for CTMCs with the same type of control parameters (see Proposition \ref{thm:ArbDist}).

Second, we have proven, that using time varying control parameters, asymptotic controllability of the system to any probability distribution is possible.

Third, we show that the density-dependent generator of a CTMC can be designed to have a decentralized structure and to converge to the zero matrix at equilibrium. This convergence of the control inputs to zero stops the agents from switching between states, and thus potentially wasting energy on unnecessary transitions, once the target distribution is reached. As pointed out in \cite{bandyopadhyay2013inhomogeneous}, the controllers developed in the prior work described above have nonzero control inputs at equilibrium, resulting in continued agent switching between states.
We present a proof by construction of our third main result for graphs of arbitrary size.  

In addition to our theoretical results, we develop an algorithm using sum-of-squares (SOS) tools to construct  density-dependent control laws with our desired properties.  Our nonlinear control approach in this work differs from our approach in \cite{elamvaz2016lin}, where we investigate linearization-based controllers for CTMCs with the same specifications.  While linear controllers have low computational complexity, they violate positivity constraints on the control inputs. To realize linear controllers in practice for our problem, we can implement them with rational feedback laws that mimic their behavior, as we show in \cite{elamvaz2016lin}. However, this approach results in unbounded controls.  In contrast, the controllers that we develop in this paper take the form of positive polynomials, and we can therefore guarantee their global boundedness. Additionally, in contrast with the approaches presented in \cite{bandyopadhyay2013inhomogeneous},\cite{demir2015decentralized} when agent dynamics are given by DTMCs, all computations for the control synthesis is done offline in our methodology. Hence, the computational burden on the agents is significantly much lower in our work in comparison.

\section{NOTATION}

We denote by $\mathcal{G = (V,E)}$ a directed graph with $M$ vertices, $\mathcal{V} = \lbrace 1,2,...,M \rbrace$, and a set of $N_{\mathcal{E}}$ edges, $\mathcal{E}\subset \mathcal{V} \times \mathcal{V}$.  We say that $e = (i,j) \in \mathcal{E}$ if there is an edge from vertex $i \in \mathcal{V}$ to vertex $j \in \mathcal{V}$. We define a source map $S : \mathcal{E} \rightarrow \mathcal{V}$  and a target map $T: \mathcal{E} \rightarrow \mathcal{V}$ for which $S(e) = i$ and $T(e) = j$ whenever $e = (i,j) \in \mathcal{E}$. There is a {\it directed path} of length $f$ from vertex $i\in \mathcal{V}$ to vertex $j\in \mathcal{V}$ if there exists a sequence of edges $\{e_k\}^f_{k=1}$ in $\mathcal{E}$ with $S(e_1) = i$, $T(e_f) = j$, and $S(e_j)=T(e_{j-1})$ for all $j \in \lbrace 2,3,...,M \rbrace$. We assume that the graph $\mathcal{G}$ is {\it strongly connected}, which means a directed path exists from any vertex $i \in \mathcal{V}$ to any other vertex $j \in \mathcal{V}$. We assume that $(i,i) \notin \mathcal{E}$ for all $i \in \mathcal{V}$. The graph $\mathcal{G}$ is said to be {\it bidirected} if $e = (S(e),T(e)) \in \mathcal{E}$ implies that $\tilde{e} = (T(e),S(e))$ also lies in $\mathcal{E}$. 

We define $\mathbb{R}^M$ as the $M$-dimensional Euclidean space, $\mathbb{R}^{M \times N}$ as the space of $M \times N$ matrices, and $\mathbb{R}_+$ as the set of positive real numbers. The notation ${\rm int}(B)$ refers to the interior of the set $B \subset \mathbb{R}^M$.
Given a vector $\mathbf{x}  \in \mathbb{R}^M$, $\mathbf{x}_i$ denotes the $i^{th}$ coordinate value of $\mathbf{x}$. 
For a matrix $\mathbf{A} \in \mathbb{R}^{M \times N}$, $\mathbf{A}^{ij} $ denotes the element in the $i^{th}$ row and $j^{th}$ column of $\mathbf{A}$. Given a vector $\mathbf{y} \in \mathbb{R}^M$, for each vertex $i\in \mathcal{V}$, the set $\sigma_\mathbf{y}(i) \subset \mathcal{V}$ consists of all vertices $j$ for which there exists a directed path $\{e_i\}^f_{i=1}$ of some length $f$ from $j$ to $i$ such that $\mathbf{y}_{S(e_k)} = 0$ for each $1 \leq k \leq f-1$.

We say that a vector $\mathbf{x}^d \in \mathbb{R}^M $ has a {\it strongly connected support} if the subgraph $\mathcal{G}_{sub} = (\mathcal{V}_{sub},\mathcal{E}_{sub})$, defined by $\mathcal{V}_{sub} = \lbrace v \in \mathcal{V}: \mathbf{x}^d_v >0 \rbrace$ and $\mathcal{E}_{sub} =  \mathcal{V}_{sub} \times \mathcal{V}_{sub}\cap \mathcal{E}$, is strongly connected. Moreover, $\mathcal{V}_{sub}$ is called the {\it support} of the vector $\mathbf{x}^d$. The matrix $\mathcal{L}_{out}(\mathcal{G}) = \mathbf{D}_{out}(\mathcal{G}) - \mathbf{A}(\mathcal{G}) \in \mathbb{R}^{M \times M}$ denotes the {\it out-Laplacian} of the graph $\mathcal{G}$, where $\mathbf{D}_{out}(\mathcal{G})$ is the out-degree matrix of $\mathcal{G}$ and $\mathbf{A}(\mathcal{G})$ is the adjacency matrix of $\mathcal{G}$. $\mathbf{D}_{out}(\mathcal{G})$ is a diagonal matrix for which $(\mathbf{D}_{out}(\mathcal{G}))^{ii}$ is the total number of edges $e$ such that $S(e)=i$. The entries of $\mathbf{A}(\mathcal{G})$
are defined as $(\mathbf{A}(\mathcal{G}))^{ij} = 1$ if $(j,i) \in \mathcal{E}$, and 0 otherwise. When the graph $\mathcal{G}$ is bidirected, $\mathcal{L}_{out}(\mathcal{G})$ is the usual Laplacian of the graph, and we will drop the subscript and denote it by $\mathcal{L}(\mathcal{G})$.

\section{PROBLEM STATEMENT} \label{section:problem statement}
Consider a swarm of $N$ autonomous agents whose states evolve in continuous time according to a Markov chain with finite state space $\mathcal{V}$. As an example application of interest, $\mathcal{V}$ can represent a set of spatial locations that are obtained by partitioning the agents' environment. The graph $\mathcal{G}$ determines the pairs of vertices (states) between which the agents can transition. We define $u_{e}:[0,\infty) \rightarrow \mathbb{R}_+$ as a {\it transition rate} for each $e=(i,j) \in \mathcal{E}$. The evolution of the $N$ agents' states over time $t$ on the state space $\mathcal{V}$ is described by $N$ stochastic processes, $X_i(t) \in \mathcal{V}$, $i=1,...,N$. Each stochastic process $X_i(t)$ evolves according to the following conditional probabilities for each $e \in \mathcal{E}$:
\begin{equation}
	\mathbb{P}(X_i(t+h) = T(e) | X_i(t) = S(e)) =~ u_{e}(t)h + o(h). \label{eq:ConditionalProb}
\end{equation}
Here, $o(h)$ is the little-oh symbol and $\mathbb{P}$ is the underlying probability measure defined on the space of events $\Omega$ (which will be left undefined, as is common) induced by the stochastic processes $\lbrace X_i(t) \rbrace_{i=1}^N$. Let $\mathcal{P(V)}$ be the set of probability densities on $\mathcal{V}$. Then $\mathcal{P(V)}$ can be associated with the $(M-1)$ dimensional simplex, $\lbrace \mathbf{y} \in \mathbb{R}^M_+: ~\sum_i y_i = 1 \rbrace$. Let $\mathbf{x}(t) \in \mathbb{R}^n$ be the vector of probability distributions of the random variable $X(t)$ at time $t$, that is,
\begin{equation}
	x_i(t) = \mathbb{P}(X(t) = i), ~~~ i \in \lbrace 1,...,M \rbrace.
\end{equation}
In the case of continuous time and countable state space, the evolution of probability distributions is determined by the \textit{Kolmogorov forward equation}. Since the $X_i(t)$ are identically distributed random variables, the forward equation can be represented by a linear system of ordinary differential equations (ODEs),
\begin{equation}
	\dot{\mathbf{x}}(t) =\mathbf{G}^T\mathbf{x}(t), \hspace{3mm}
	\mathbf{x}(0) \in \mathcal{P(V)},
	\label{eq:OLsys1}
\end{equation}
where the matrix $\mathbf{G}$ is the $M \times M$ generator of the process. Each element $G^{ij}$, where $(i,j) = e \in \mathcal{E}$, is the probability per unit time, defined as the transition rate $u_e$ in Equation \eqref{eq:ConditionalProb}, of an agent switching from vertex $i = S(e)$ to vertex $j = T(e)$. The number of transitions between vertices $i$ and $j$ in $h$ units of time has a Poisson distribution with parameter $G^{ij} h$; see \cite{norris1998markov} for details. The elements of $\mathbf{G}$ have the following properties:
\begin{equation}
	0 \leq G^{ij}  < \infty, \hspace{3mm} G^{ii} = -\sum\limits_{j=1,j\neq i}^M G^{ij}.
\end{equation}

The system (\ref{eq:OLsys1}) can be cast in an explicitly control-theoretic form,
\begin{equation}
	\dot{\mathbf{x}}(t) =\sum\limits_{e \in \mathcal{E}}  u_e \mathbf{B}_{e} \mathbf{x}(t), \hspace{3mm} \mathbf{x}(0) \in \mathcal{P(V)}, \\
	\label{eq:OLsys2}
\end{equation}
where $\mathbf{B}_e$, $e \in \mathcal{E}$, are control matrices with entries
\[
B_e^{ij} = 
\begin{cases} 
-1 & \text{if } i = j=  S(e),\\
1 & \text{if } i= T(e), \hspace{1mm} j = S(e),\\
0       & \text{otherwise.}
\end{cases}
\]

The focus of this paper is to solve the problem of achieving arbitrary distributions using density feedback control. For clarity, we first consider the open-loop version of our control problem, before moving on to the closed-loop version. Given a desired probability distribution $\mathbf{x}^d$, the problem of computing  the transition rates (control parameters) $\lbrace u_e \rbrace_{e \in \mathcal{E}}$ to achieve the desired distribution can be framed as follows:
\begin{problem}
	Find positive control parameters $\lbrace u_e \rbrace_{e \in \mathcal{E}}$ such that $\lim_{t \rightarrow \infty} \| \mathbf{x}(t) - \mathbf{x}^d \| = 0  $ for all $\mathbf{x}^0 \in \mathcal{P(V)}$.
	\label{prob:ArbDist}
\end{problem}
We provide a complete characterization of the stationary distributions that are stabilizable for this case. Although density- and time-independent transition rates of CTMCs have been previously computed in an optimization framework \cite{berman2009optimized}, the question of which equilibrium distributions are feasible has remained unresolved for the case where the target distribution is not strictly positive on all vertices. While only strictly positive target distributions have been considered in previous work on control of swarms governed by CTMCs \cite{berman2009optimized,halasz2007dynamic}, we address the more general case in which the target densities of some states can be zero. This question was addressed in \cite{acikmese2012markov} for swarms governed by DTMCs. The problem has also been investigated in the context of consensus protocols \cite{chapman2015advection} for strictly positive distributions, where what is referred to as ''advection on graphs'' is in fact the forward equation corresponding to a CTMC. In our controller synthesis, we will relax the assumption of strict positivity for desired target distributions.

The main problem that we address in this paper is the following:
\begin{problem}
	Given a strictly positive desired equilibrium distribution $\mathbf{x^d} \in \mathcal{P(V)}$, compute transition rates $u_e: \mathcal{P(V)} \rightarrow\mathbb{R}_+$, $e \in \mathcal{E}$, such that the closed-loop system
	\begin{equation}
	\dot{\mathbf{x}}(t) =\sum_{e\in\mathcal{E}} u_e(\mathbf{x}) \mathbf{B}_e \mathbf{x}(t)
	\label{eq:CLsys1}
	\end{equation}
	satisfies $\lim_{t \rightarrow \infty} \| \mathbf{x}(t) - \mathbf{x}^d \| = 0  $ for all $\mathbf{x}^0 \in \mathcal{P(V)}$, with the additional constraint that $u_e(\mathbf{x}^d)=0$ for all$e \in \mathcal{E}$. Moreover, the density feedback should have a decentralized structure, in that each $u_e$ must be a function only of densities $x_i$ for which $i = S(e)$ or $i = S(\tilde{e})$, where $T(\tilde{e}) = S(e)$.
	\label{prob:Control}
\end{problem}


We note that we were able to describe the state evolution of the agents by system \eqref{eq:OLsys2} when the transition rates were density-independent because the agents' states were independent and identically distributed (i.i.d.) random variables in that case. However, when the density feedback control law 
$\{u_e(\mathbf{x})\}_{e \in \mathcal{E}}$ is used, the independence of the stochastic processes $X_i(t)$ is lost. This implies that the evolution of the probability distribution cannot be described by system \eqref{eq:OLsys2}. However, if we invoke the \textit{mean-field hypothesis} and take the limit $N \rightarrow \infty$, then we can model the evolution of the probability distribution according to a nonlinear Markov chain. In this limit, the number of agents at vertex $v \in \mathcal{V}$ at time $t \in [0, T]$ where, $T>0$, denoted by  $N_v(t,\omega)$ (where $\omega$ is used to emphasize that $N_v(\cdot)$ is a random variable, which denotes the number of agents in vertex $\mathcal{(V)}$), converges to $x_v(t)$ in an appropriate sense, provided that solutions of \eqref{eq:CLsys1} are defined until a given final time $T>0$. 
A rigorous process for taking this limit in a stochastic process setting is described in \cite{ethier2009markov,kolokoltsov2010nonlinear}.  




\section{ANALYSIS}
In this section, we first address the controllability problem in Problem \ref{prob:Control}, and then the stabilizability problem in Problem \ref{prob:ArbDist}.

\subsection{Controllability}
\begin{proposition}
	Let $\mathcal{G}$ be a strongly connected graph. Suppose that $\mathbf{x}^0 \in \mathcal{P(V)}$ is an initial distribution and $\mathbf{x}^d \in \mathcal{P(V)}$ is a desired distribution. Additionally, assume that $\mathbf{x}^d$ has strongly connected support. Then there is a set of parameters, $a_e \in [0, \infty)$ for each $e \in \mathcal{E}$, such that if $u_e(t)= a_e$ for all $t \in [0, \infty)$ and for each $e \in \mathcal{E}$ in system \eqref{eq:CLsys1}, then the solution $\mathbf{x}(t)$ of this system satisfies $ \| \mathbf{x}(t) - \mathbf{x}^d \| \leq M e^{-\lambda t}  $ for all $t \in [0, \infty)$ and for some positive parameters $M$ and $\lambda$ that are independent of $\mathbf{x}^0$.
	\label{thm:ArbDist} 
\end{proposition}

\begin{proof}
	Let $\mathcal{V}_s \subset \mathcal{V}$ be the support of $\mathbf{x}^d$. From this vertex set, we construct a new graph $\tilde{\mathcal{G}}= (\mathcal{V},\tilde{\mathcal{E}})$, where $e=(i,j) \in \mathcal{E}$ implies that $e \in \tilde{\mathcal{E}} $ if and only if $i \in \mathcal{V}_s$ implies that $j \not\in \mathcal{V} \backslash \mathcal{V}_s$. Then it follows from \cite{chapman2015advection}[Proposition 10] that the process generated by the transition rate matrix $-\mathcal{L}_{out}(\mathcal{\tilde{G})}^T$ has a unique, globally stable invariant distribution if we can establish that $\tilde{\mathcal{G}}$ has a {\it rooted in-branching} subgraph.  This  implies that $\tilde{\mathcal{G}}$ must have a subgraph $\tilde{\mathcal{G}}_{sub} =(\mathcal{V},\mathcal{E}_{sub})$ which has no directed cycles and for which there exists a root node, $v_r$, such that for every $v \in \mathcal{V}$  there exists a directed path from $v$ to $v_r$. This is indeed true for the graph $\tilde{\mathcal{G}}$, which can be shown as follows. First, let $r \in \mathcal{V}$ such that $x^d_r > 0 $. From the assumption that $\mathcal{G}$ is strongly connected and the construction of $\tilde{\mathcal{G}}$, it can be concluded that there exists a directed path in $\tilde{\mathcal{E}}$ from any $v\in \mathcal{V}$ to $r$. Now, for each $n \in  \mathbb{Z}_+$, let $\mathcal{N}_n(r)$ be the set of all vertices for which there exists a directed path of length $n$ to $r$. For each $n >1$, let $\tilde{\mathcal{N}}_n(r) = \mathcal{N}_n(r) \backslash \cup_{m=1}^{n-1} \mathcal{N}_m(r)$. 
	We define $\tilde{\mathcal{E}}_{sub}$ by setting $e \in \tilde{\mathcal{E}}_{sub} $ if and only if $e \in \mathcal{E}$, $S(e) \in \tilde{\mathcal{N}}_n(r)$, and $T(e) \in \tilde{\mathcal{N}}_{n-1}(r)$ for some $n >1$. Then $\tilde{\mathcal{G}}_{sub} =(\mathcal{V},\mathcal{E}_{sub})$ is the desired rooted in-branching subgraph.
	
	The matrix $-\mathcal{L}_{out}(\mathcal{\tilde{G})}^T$ is the generator of a CTMC, since $\mathcal{L}_{out}(\mathcal{\tilde{G})}^T \mathbf{1} = \mathbf{0}$ and its off-diagonal entries are positive. Moreover, as we have shown, $\mathcal{\tilde{G}}$ has a rooted in-branching subgraph. Hence, there exists a unique vector $\mathbf{z}$ such that $-\mathcal{L(\tilde{G})}\mathbf{z} = \mathbf{0}$ and $\mathbf{z}  \in \mathcal{P}(\mathcal{V})$. The vector $\mathbf{z}$ is nonzero only on $\mathcal{V}_s$, since the subgraph corresponding to $\mathcal{V}_s$ is strongly connected. Then we consider a positive definite diagonal matrix $\mathbf{D} \in \mathbb{R}^{M \times M}$ such that $D^{ii}= z_i/x^d_i$ if $i \in \mathcal{V}_s$ and an arbitrary strictly positive value for any other $i \in \mathcal{V}$. The matrix $-\mathbf{D}\mathcal{L}_{out}(\mathcal{\tilde{G})}^T$ is also the generator of a CTMC. Moreover, $\mathbf{x}^d$ is the unique stationary distribution of the process generated by $-\mathbf{D}\mathcal{L}_{out}(\mathcal{\tilde{G})}^T$, since $\mathbf{x}^d$ lies in the null space of $\mathbf{G} = -\mathcal{L}_{out}(\mathcal{\tilde{G})} \mathbf{D}$ by construction. The simplicity of the principal eigenvalue at $0$ for the matrix $-\mathbf{D}\mathcal{L}_{out}(\mathcal{\tilde{G})}^T$ is inherited by the same eigenvalue of the matrix $\mathbf{G}$. Then the result follows by setting $a_e = G^{T(e)S(e)}$ for each $e \in \mathcal{E}$ and by noting that since $\mathbf{G}^T$ is the generator of a CTMC, and its eigenvalue at zero has the aforementioned properties is simple, then the rest of the spectrum of $\mathbf{G}$ lies in the open left half of the complex plane.
\end{proof}

This result can be extended to the case of time-varying control parameters. In particular, any $\mathbf{x}^d \in {\rm int}(\mathcal{P})$ can be reached in finite time from a given $\mathbf{x}^0 \in \mathcal{P}$ using time-varying control parameters $\lbrace u_e(t) \rbrace_{e \in \mathcal{E}}$. We restate the following theorem from our companion paper \cite{elamvaz2016lin}.
\begin{theorem}
	\label{ctrtheo}
	\cite{elamvaz2016lin}
	If the graph $\mathcal{G} = (\mathcal{V}, \mathcal{E})$ is strongly connected,
	then the system \ref{eq:OLsys2} is small-time globally controllable
	from every point in the interior of the simplex defined by $\mathcal{P}(\mathcal{V})$.
\end{theorem}

\begin{remark}
	In fact, we can state the following broader result. If $\mathcal{G}$ is strongly connected, then the system is also {\it path controllable}: given any trajectory $\gamma(t)$ in $
	\mathcal{P}(\mathcal{V})$ that is defined over a finite time interval [0, T] and is once differentiable with respect to the time variable $t$, there exists a control law $\mathbf{u}:[0,T] \rightarrow [0,\infty)^{N_\mathcal{E}}$ such that the solution of the control system \eqref{eq:CLsys1} satisfies $\mathbf{x}(t) = \gamma(t)$ for all $t \in [0,T]$. This is true because conical combinations of the collection of vectors $\lbrace \mathbf{B}_e\mathbf{y} \rbrace_{e \in \mathcal{E}}$ span the tangent space of $\mathcal{P}(\mathcal{V})$ whenever $\mathbf{y}$ lies in the interior of $\mathcal{P}(\mathcal{V})$. For example, given a strongly connected graph $\mathcal{G}$, if $(i,j) \in \mathcal{E}$ and there exists a directed path $\mu$ from $j$ to $i$, then $-\mathbf{B}_{(i,j)} \mathbf{1} = \sum_{e \in \mu}\mathbf{B}_e \mathbf{1}$.
\end{remark}

As we mention in \cite{elamvaz2016lin}, this result cannot be extended to prove reachability of distributions that correspond to points on the boundary of $\mathcal{P(V)}$. On the other hand, the following theorem states that these boundary points are {\it asymptotically controllable}. A key difference between the following result and the results in Proposition \ref{thm:ArbDist} and Theorem \ref{ctrtheo} is that the target distributions need not have strongly connected supports.

\begin{proposition}
	Let $\mathcal{G}$ be a strongly connected graph. Suppose that $\mathbf{x}^0 \in \mathcal{P(V)}$ is the initial distribution, and $\mathbf{x}^d \in \mathcal{P(V)}$ is the desired distribution. Then for each $e \in \mathcal{E}$, there exists a set of time-dependent control parameters $u_e : \mathbb{R}_+ \rightarrow \mathbb{R}_+$, $e \in \mathcal{E}$, such that the solution $\mathbf{x}(t)$ of the controlled ODE \eqref{eq:CLsys1} satisfies $\lim_{t \rightarrow \infty} \mathbf{x}(t) = \mathbf{x}^d$.
\end{proposition}

Before presenting the full proof of this proposition, we briefly sketch the proof for clarity. The proof is mainly based on Theorem \ref{ctrtheo} and uses an approach similar to that used in Proposition \ref{thm:ArbDist} to prove the existence of control inputs that stabilize desired distributions $\mathbf{x}^d$ with strongly connected support. The idea is to first partition the vertex set $\mathcal{V}$ into disjoint subsets $\lbrace \mathcal{V}_i \rbrace$ that each contain a single vertex $r$, called the {\it root node}, for which $x^d_r \neq 0$,  
and some other vertices $v$,  called {\it transient nodes}, for which $x^d_v = 0$ 
and there is a directed path to the root node in $\mathcal{V}_i$. 
This partition will ensure that the subgraphs corresponding to $\lbrace \mathcal{V}_i \rbrace$ are at least weakly connected. Then using Theorem \ref{ctrtheo}, we can design control inputs that drive the solutions of the system \eqref{eq:OLsys2} exactly to an intermediate distribution $\mathbf{x}^{in}$, for which the total mass at the vertices in $\mathcal{V}_i$ is equal to the total mass required at the root node in $\mathcal{V}_i$. Such a distribution $\mathbf{x}^{in}$ necessarily exists in ${\rm int}(\mathcal{S})$. Then we invoke an argument made in Proposition \ref{thm:ArbDist} to ensure that all the mass at the transient nodes is directed toward their corresponding root nodes, which can be achieved using time-variant control inputs. This will establish asymptotic controllability of the boundary points of $\mathcal{P}\mathcal{(V)}$.


\begin{proof}
	We define the set $\mathcal{R} = \lbrace i: \hspace{1mm} x^d_i > 0, \hspace{2mm} 1 \leq i \leq M \rbrace$ with cardinality $N_{\mathcal{R}}$. Let $\mathcal{I}:\lbrace 1,2,...,N_\mathcal{R} \rbrace \rightarrow \mathcal{R} $  be a bijective map that 
	defines an ordering on $\mathcal{R}$. Then we recursively define a collection $\lbrace{\mathcal{V}_n} \rbrace$ of disjoint subsets of $\mathcal{V}$ as follows:
	\begin{align}
	\mathcal{V}_1 = \lbrace \mathcal{I}(1)\} \cup \lbrace i \in \mathcal{V}: \hspace{1mm} x^d_i = 0 \hspace{2mm} s.t. \hspace{2mm} i \in \sigma_{ \mathbf{x}^d}(\mathcal{I}(1))  \rbrace \nonumber \\
	\mathcal{V}_n = \lbrace \mathcal{I}(n)\rbrace \cup \lbrace i \in \mathcal{V}: \hspace{1mm} x^d_i = 0 \hspace{2mm} s.t. \hspace{2mm} i \in \sigma_{ \mathbf{x}^d}(\mathcal{I}(n)) \hspace{2mm} \nonumber \\ 
	and \hspace{2mm} i \notin \cup_{k=1}^{n-1}\mathcal{V}_k)  \rbrace \nonumber
	\end{align}
	for each $n \in \lbrace 2,3,...,N_{\mathcal{R}}\rbrace$. We note that $\mathcal{V} = \cup_{n=1}^{N_{\mathcal{R}}}\mathcal{V}_{n}$. Let $\mathbf{x}^{in} \in {\rm int}(\mathcal{P(V)})$ be some element such that $\sum_{k \in \mathcal{V}_n}x^{in}_k = x^d_{\mathcal{I}(n)}$ for each $n \in \lbrace 1,2,...,N_{\mathcal{R}} \rbrace$. From Theorem \ref{ctrtheo}, we know that there exists a control $u^1_e:[0,T] \rightarrow \mathbb{R}_+$ for each $e \in \mathcal{E}$ such that the solution $\mathbf{x}(t)$ of the control system \eqref{eq:OLsys2} satisfies $\mathbf{x}(T)= \mathbf{x}^{in}$. Now we will design $\lbrace u_e\rbrace_{e \in \mathcal{E}}$ such that $u_e(t) = u^1_e(t)$ for each $t \in [0,T]$ and $u_e(t) = a_e$ for each $t \in (T,\infty]$, where $a_e$ is defined as follows:
	\[
	a_e= 
	\begin{cases} 
	0 &\text{if } S(e) \in \mathcal{V}_ n \text{ and } T(e) \notin \mathcal{V}_n \hspace{0.5mm} ~~\forall 1 \leq n \leq N_{\mathcal{R}},   \\
	0 &\text{if } S(e)=\mathcal{I}(n) ~\text{ for some } 1 \leq n \leq N_{\mathcal{R}},  \\
	1 & \text{otherwise.}
	\end{cases}
	\]
	Then the solution of system \eqref{eq:OLsys2} for $t>T$ can be constructed from the solution of the following decoupled set of ODEs:
	\begin{eqnarray}
	\label{eq:Asymsys}
	\dot{\mathbf{y}}_n(t) &=& -\mathcal{L}_{out}(\mathcal{\tilde{G}}_n)\mathbf{y}_n(t), \hspace{3mm} t \in [T, \infty) \\ \nonumber
	\mathbf{y}_n(T) &=& \mathbf{y}_n^0 \in \mathcal{P}(\mathcal{V}_n) 
	\end{eqnarray}
	for $1 \leq n \leq N_{\mathcal{R}}$. Here, $\mathcal{G}_n=(\mathcal{V}_n,\mathcal{E}_n)$ for each $1 \leq n \leq N_{\mathcal{R}}$, where $e \in \mathcal{E}_n$ if $S(e),T(e) \in \mathcal{V}_n$, and $a_e =1$. The solution of system \eqref{eq:Asymsys} is related to the solution of system \eqref{eq:OLsys2} with $\mathbf{x}(T) = \mathbf{x}^{in}$ through a suitable permutation matrix $\mathbf{P}$: $\mathbf{P}\mathbf{x}(t) = [\mathbf{y}_1(t) ~ \mathbf{y}_2(t) ~....
	~\mathbf{y}_{N_{\mathcal{R}}}(t)]$. Since each graph $\mathcal{G}_n$ has a rooted in-branching subgraph, the process generated by $ -\mathcal{L}_{out}(\mathcal{\tilde{G}}_n)^T$ has a unique stationary distribution. Moreover, by construction, this unique, globally stable stationary distribution is the vector $[x^d_{\mathcal{I}(n)} ~ \mathbf{0}_{1 \times (|\mathcal{V}_n|-1)}]^T$, where $|\mathcal{V}_n|$ is the cardinality of the set $\mathcal{V}_n$. This implies that $\lim _{t \rightarrow \infty}\mathbf{P}^{-1}\mathbf{y}(t) =\lim _{t \rightarrow \infty} \mathbf{x}(t) = \mathbf{x}^d$. By concatenating the control inputs $\lbrace u^1_e \rbrace_{e \in \mathcal{E}}$ and $\lbrace a_e \rbrace_{e \in \mathcal{E}}$, we obtain the desired asymptotic controllability result.
\end{proof}

An interesting aspect of the above proof is its implication that asymptotic controllability is achievable with piecewise constant control inputs with a finite number of pieces. From the above result, it follows that any point in $\mathcal{P}(V)$ can be stabilized using a full-state feedback controller \cite{clarke1997asymptotic}. However, for a general target equilibrium distribution, a stabilizing controller with a decentralized structure might not exist.

Before we present an algorithm to construct polynomial feedback control laws, it is important that we address the feasibility of Problem \ref{prob:Control}. Toward this end, we will investigate the stabilizability of the system (\ref{eq:OLsys2}).

\subsection{Stabilizability} \label{section:stabilize}
We will prove stabilizability by constructing an explicit control law for a graph of arbitrary size that fulfills all the conditions of Problem \ref{prob:Control}.
We propose the following decentralized control law, which depends on the agent densities in different states, and prove that the resulting closed-loop system is asymptotically stable. 
For $i,j \in \lbrace 1,...,M\rbrace$, let $g_i(x_i(t)) = (x_i(t)-x_i^d)^2$, and let $w_{ij}=1$ if $(i,j)\in \mathcal{E}$ and $0$ otherwise. Define a transition rate (control) matrix $\mathbf{G}(\mathbf{x}(t))$ with the following entries:
\begin{eqnarray}
	G^{ij} = 
	\begin{cases} 
		~~w_{ij}(g_i+g_j), & \hspace{3mm} i\neq j \nonumber \\
		-\sum\limits_{k=1}^M w_{1k}(g_i+g_k), & \hspace{3mm} i = j. 
	\end{cases} \label{eq:control} \\
\end{eqnarray}
$\mathbf{G}$ thus defined satisfies all the properties of a transition rate matrix described in Section \ref{section:problem statement}; that is, each row sums to 1 and each element is non-negative. It is clear that when $x_i(t)=x_i^d$ for all $i \in \mathcal{V}$, then all $g_i = 0$, resulting in $\mathbf{G}(\mathbf{x}^d)=\mathbf{0}$, which satisfies our requirement that the control parameters equal zero at equilibrium. The second requirement of a decentralized control structure is enforced by setting $w_{ij} = 0$ whenever $(i,j) \neq \mathcal{E}$. All that remains is to prove that, with this choice of $\mathbf{G}$, the closed-loop system is asymptotically stable.

\begin{proposition}
	The closed-loop system
	\begin{equation}
		\dot{\mathbf{x}}(t)=\mathbf{G}(\mathbf{x}(t))^T\mathbf{D}\mathbf{x}(t)
		\label{eq:CLsys2}
	\end{equation}
	with $\mathbf{G}$ defined as in Equation (\ref{eq:control}) is asymptotically stable. 
\end{proposition}

\begin{proof} 
	For ease of representation, we will use this system description rather than the equivalent system (\ref{eq:CLsys1}).
	To prove the stability of this system, we propose the following candidate Lyapunov function:
	\begin{equation}
		V(\mathbf{x}) = \frac{1}{2}\left(\mathbf{x}(t)^T\mathbf{D}\mathbf{x}(t)-(\mathbf{x}^d)^T \mathbf{D} \mathbf{x}^d\right).
		\label{eq:V}
	\end{equation}
	
	We now check the conditions for this function to be a Lyapunov function for the desired equilibrium point $\mathbf{x}^d$. We clearly have that $V(\mathbf{x}^d)=0$. To prove that $V(\mathbf{x})>0$ for all $\mathbf{x}\in \mathcal{P(V)}\backslash\lbrace \mathbf{0} \rbrace$, we note the following:
	\begin{align}
		V(\mathbf{x}) &= \frac{1}{2}\left(\mathbf{x}(t)^T\mathbf{D}\mathbf{x}(t)-(\mathbf{x}^d)^T \mathbf{D} \mathbf{x}^d\right) \nonumber\\
		& = \frac{1}{2}\left((\mathbf{D}^\frac{1}{2} \mathbf{x})^T (\mathbf{D}^\frac{1}{2} \mathbf{x}) -1\right) \nonumber\\
		&= \frac{1}{2}\left(\langle \mathbf{D}^\frac{1}{2} \mathbf{x},\mathbf{D}^\frac{1}{2} \mathbf{x} \rangle-1\right). \label{eq:LyapPos}
	\end{align}
	We will now show that the minimum value that $\langle\mathbf{D}^\frac{1}{2} \mathbf{x},\mathbf{D}^\frac{1}{2} \mathbf{x} \rangle$ can attain on $\mathcal{P(V)}$ is $1$, which is possible only at $\mathbf{x}^d$, guaranteeing strict positivity of the expression (\ref{eq:LyapPos}) for any other $\mathbf{x} \in \mathcal{P(V)}$. We apply the following coordinate transformation to shift the simplex associated with $\mathcal{P(V)}$, so that $\mathbf{x}^d$ coincides with the origin. Let $\mathbf{y}=\mathbf{x}-\mathbf{x}^d$. Then, 
	\begin{equation}
		\sum_{i=1}^M y_i = \sum_{i=1}^M (x_i-x^d) = 0 \label{(eq:ysum)},
	\end{equation}
	and therefore,
	\begin{align}
		\langle \mathbf{D}^\frac{1}{2} \mathbf{x},\mathbf{D}^\frac{1}{2} \mathbf{x} \rangle &= \langle \mathbf{D}^\frac{1}{2} (\mathbf{y+x}^d),\mathbf{D}^\frac{1}{2} (\mathbf{y+x}^d) \rangle \nonumber \nonumber \\
		&= \langle \mathbf{y},\mathbf{D} \mathbf{y} \rangle + 2\langle \mathbf{y},\mathbf{D} \mathbf{x}^d \rangle + \langle \mathbf{x}^d,\mathbf{D} \mathbf{x}^d \rangle \nonumber \\
		&=  \langle \mathbf{y},\mathbf{D} \mathbf{y} \rangle + 1
	\end{align}
	
	
	Since $\mathbf{Dx}^d=\mathbf{1}$ and $\langle \mathbf{y},\mathbf{1}\rangle=0$ (this follows from Equation (\ref{(eq:ysum)})), the function (\ref{eq:LyapPos}) is positive on all $\mathbf{x}\in \mathcal{P(V)}\backslash\lbrace \mathbf{0} \rbrace$. 
	
	Lastly, we compute the time derivative of the candidate Lyapunov function:
	\begin{align}
		\dot{V}(\mathbf{x}(t)) &= \frac{1}{2}\dot{\mathbf{x}}(t)^T\mathbf{Dx}(t) + \frac{1}{2}\mathbf{x}(t)^T\mathbf{D}\dot{\mathbf{x}}(t) \nonumber \\
		&= \frac{1}{2}(\mathbf{G}^T\mathbf{Dx}(t))^T\mathbf{Dx}(t) + \frac{1}{2}\mathbf{x}^T(t)\mathbf{D}(\mathbf{G}^T\mathbf{D}\mathbf{x}(t)) \nonumber \\
		&= \mathbf{x}(t)^T(\mathbf{DGD})\mathbf{x}(t). 
		\label{eq:Vdot}
	\end{align}
	For the equilibrium $\mathbf{x}^d$ to be asymptotically stable, we must have $\dot{V}(\mathbf{x}(t))<0$, $\forall \mathbf{x}\in \mathcal{P(V)} \backslash \lbrace 0 \rbrace $. Negative semi-definiteness of  $\dot{V}$ is guaranteed by the fact that $\mathbf{G}(\mathbf{x}(t))$ is a transition rate matrix. Strict negativity of  $\dot{V}$ can be confirmed by algebraic manipulation of expression (\ref{eq:Vdot}) as follows.
	Setting $r(t) = x(t)/x^d$, we obtain:
	\begin{align}
		\dot{V}(\mathbf{x}(t)) &= (\mathbf{Dx}(t))^T \mathbf{G}(\mathbf{x}(t)) (\mathbf{Dx}(t)) \nonumber \\
		&= \mathbf{r}(t)^T\mathbf{G}(\mathbf{x}(t))\mathbf{r}(t) \nonumber\\
		&= \sum\limits_{i,j=1,i\neq j}^M -(r_i-r_j)^2 w_{ij}(g_i+g_j).
		\label{eq:Vdot2}
	\end{align}
	The expression (\ref{eq:Vdot2}) is a negative sum-of-squares (SOS) and thus equals zero only when $r_i=r_j$ for all $i, j$, which is possible only at $\mathbf{x}(t) = \mathbf{x}^d$. Hence, this function is strictly negative for all $\mathbf{x}\in \mathcal{P(V)}\backslash\lbrace \mathbf{0} \rbrace$.
	
	In summary, the function (\ref{eq:V}) fulfills all the criteria of a Lyapunov function, thus proving asymptotic stability of the the closed-loop system \eqref{eq:CLsys2}.
\end{proof}

\section{COMPUTATIONAL APPROACH}

In this section, we briefly discuss how decentralized nonlinear controls can be constructed algorithmically. By describing an algorithmic procedure, we hope to demonstrate that additional constraints can be added, to improve the performance of the closed loop system. We will construct control laws that are polynomial functions of the state of the system. We will take the aid of SOSTOOLS, short for Sum-of-Squares toolbox, used for polynomial optimization. SOSTOOL has been a very popular method to provide algorithmic solution of problems that can be formulated as polynomial non-negative constraints that are otherwise difficult to solve \cite{prajna2002introducing}. In this methods non-negativity constraint is relaxed to the existence of a SOS decomposition, which is then tested using Semidefinite programming. A point to be noted here is that the procedure described below is one of the possible methods to construct such control laws. We now pose Problem \ref{prob:Control} as an optimization problem.

\begin{problem}
Let, 
\begin{equation}
\mathcal{P(V)}= \lbrace(x_1,...,x_n) \in \mathbb{R}^n | x_i \geq 0, \hspace{1mm} \sum\limits_{i=1}^n x_i = 1, \forall i \rbrace. 
\label{eq:Simplex}
\end{equation}
Let $\mathbb{R}[x]$ represent the set of polynomials and $\Sigma_s$ denote the set of SoS polynomials. \\
Consider the system (\ref{eq:CLsys2}), which is of the form $\dot{x}=g(x)u(t)$, $u(t)=k(x)$ \\
Given, matrix $\mathbf{B}_i$ and Lyapunov function $V(\mathbf{x})$.\\
Find, $u(x) \in \mathbb{R}[x]$ such that,
\begin{align}
u (x) \geq 0 \\
u(x^d) = 0 \\
\nabla V(x)^T g(x)k(x) \leq 0 \label{eq:GradV}
\end{align}
for all $x \in \mathcal{P(V)}$
\end{problem}

Here, we are using the same Lyapunov function (\ref{eq:V}) used in Section (\ref{section:stabilize}) to prove stabilizability. We have already established that it has zero magnitude at equilibrium $\mathbf{x}^d$ and is positive everywhere on $\mathcal{P(V)}$. Hence, we only need to test for its gradient's negative definiteness, which is being encoded here. In this construction we are fixing the candidiate Lyapunov function and constructing a control law such that the \ref{eq:V} is indeed a Lyapunov function for the closed loop system \ref{eq:CLsys1}. Alternatively, one could search for both the Lyapunov function and the control law together, but this renders the problem bilinear in the 2 variables. Iterating between the two variables is one way to get around this problem.

To implement (\ref{eq:GradV}), that is, to show local negative definiteness of the gradient (on the simplex $\mathcal{P(V)}$), we use the following result well known in literature on \textit{positivestellansatz}, known as \textit{Schmudgen's} positivestellansatz, \cite{schweighofer2005optimization}.
\begin{theorem}
Suppose $S=\lbrace x:g_i(x) \geq 0, \hspace{1mm} h_i(x) =0\rbrace$ is compact. If $f(x) \geq 0$ for all $x \in S$, then there exist $s_i, r_{ij},...\in \Sigma_s$ and $t_i \in \mathbb{R}[x]$ such that,
\begin{align}
f = &1+ \sum_j t_jh_j +s_0 + \sum_i s_ig_i + \sum_{i \neq j}r_{ij}g_ig_j +\nonumber \\ 
&\sum_{i \neq j \neq k}r_{ijk}g_i g_j g_k + ...
\end{align}
\end{theorem}
 
This theorem gives sufficient conditions for positivity of the function $f$ on a semi-algebraic set (\ref{eq:Simplex}). In our case, this translates to looking for $t_i \in \mathbb{R}[x]$ and $s_i \in \Sigma_s$ such that
\begin{align}
-\frac{\partial V}{\partial t}  = \frac{\partial V}{\partial x}f-(th+s_0+\Sigma_i s_i g_i)
\end{align}
where $f$ is the vector field, $h$ is the equality constraint in (\ref{eq:Simplex}) and $g_i$ are the combinations of the inequalities in (\ref{eq:Simplex}).

\section{NUMERICAL SIMULATIONS}

We computed two types of feedback controllers for the closed-loop system (\ref{eq:OLsys2}) to redistribute populations of $N=100$ and $N=1000$ agents on the five-vertex chain graph in Fig. \ref{fig:graph}.  The first controller ({\it Case 1}) was computed using SOSTOOLS, as described in the previous section, and the second controller ({\it Case 2}) was defined according to Equation \eqref{eq:control}.
In both cases, the initial distribution was $\mathbf{x}^0 = [0.4 \hspace{2mm} 0.1 \hspace{2mm} 0.05 \hspace{2mm} 0.35 \hspace{2mm} 0.1]^T$, and the desired distribution was $\mathbf{x}^d = [0.1 \hspace{2mm} 0.2 \hspace{2mm} 0.25 \hspace{2mm} 0.4 \hspace{2mm} 0.05]^T$. 

The solution of the mean-field model with each of the two controllers and the trajectories of a corresponding stochastic simulation are compared in Figures (\ref{fig:ConstructControl_1000})-(\ref{fig:SOSControl_100}).  To speed up the convergence rate to equilibrium, all the controller gains were multiplied by a factor of 10. Also, for ease of comparison, the ODE solutions were scaled by the number of agents.  We  observe that the performance of the {\it Case 1} controller is better than that of the {\it Case 2 controller}.   We note that if faster convergence to the equilibrium is desired, this could be encoded as constraint in SOSTOOLS.  As discussed in Section \ref{section:problem statement}, the underlying assumption of using the mean-field model \eqref{eq:OLsys2} is that the  swarm behaves like a continuum. That is, the ODE \eqref{eq:OLsys2} is valid as number of agents $N \rightarrow \infty$. Hence, it is imperative to check the performance of the feedback controller for different agent populations.  We observe that the the stochastic simulation follows the ODE solution quite closely in all four simulations.  In addition, in all simulations, the numbers of agents in each state remain constant after some time; in the case of 100 agents, the fluctuations stop earlier than in the case of 1000 agents. This is due to the property of the feedback controllers that as the agent densities approach their desired equilibrium values, the transition rates tend to zero. This effect is shown explicitly in Fig. \ref{fig:AgentPath}, which plots the time evolution of a two agents' state (vertex number) during a stochastic simulation with both of the controllers. For both controllers, the agent's state remains constant after a certain time.

\begin{figure}
	\centering
	\includegraphics[width= \linewidth]{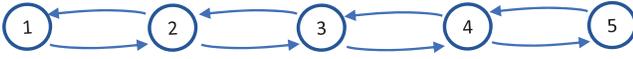}
	\caption{Five-vertex bidirected chain graph.}
	\label{fig:graph}
\end{figure}

\begin{figure}
	\centering
	\includegraphics[width= \linewidth]{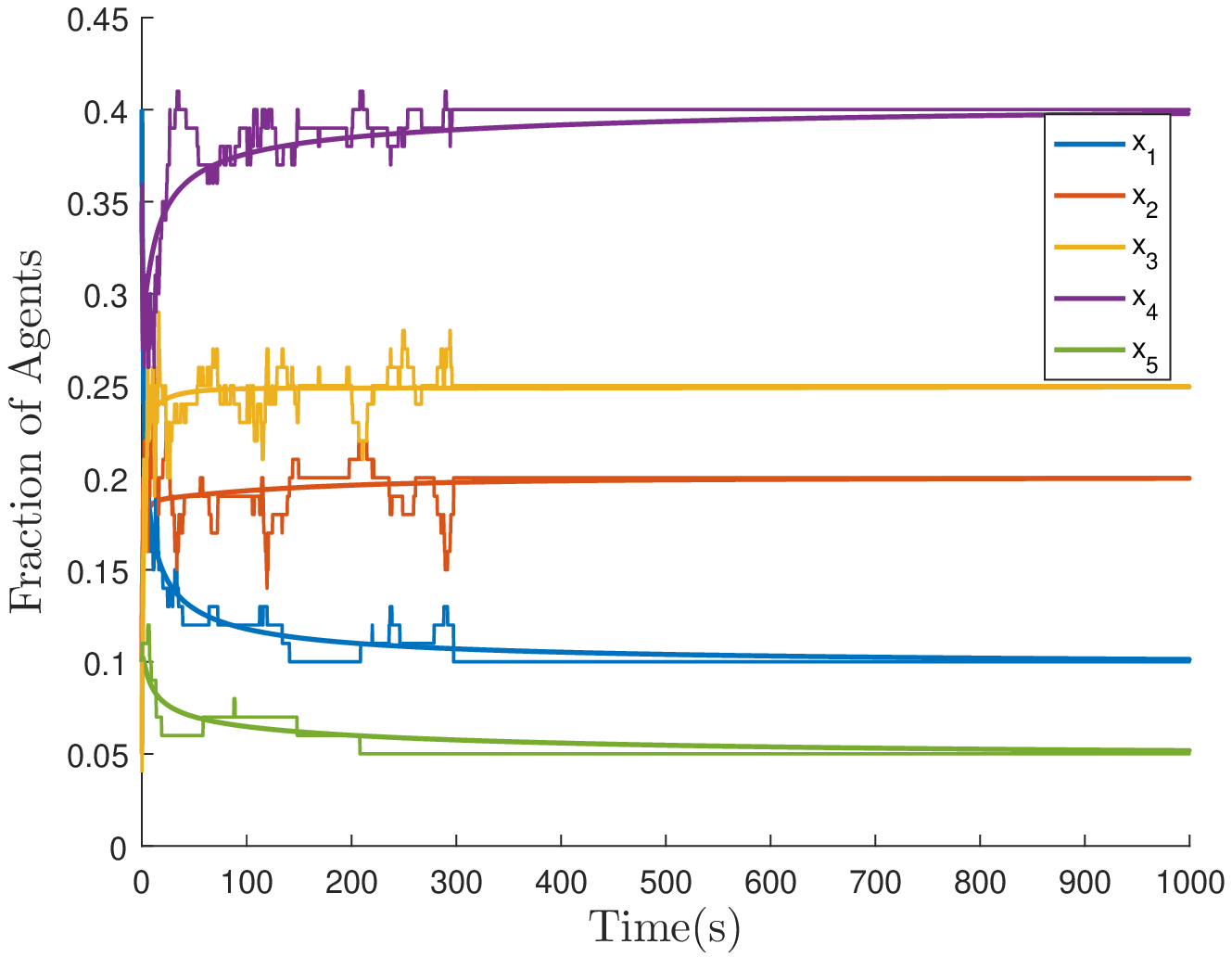}
	\caption[caption]{Trajectories of the mean-field model {\it (thick lines)} and the corresponding stochastic simulation {\it (thin lines)} for the {\it Case 1} closed-loop controller with $N=100$ agents.}
	\label{fig:ConstructControl_1000}
\end{figure}

\begin{figure}
	\centering
	\includegraphics[width= \linewidth]{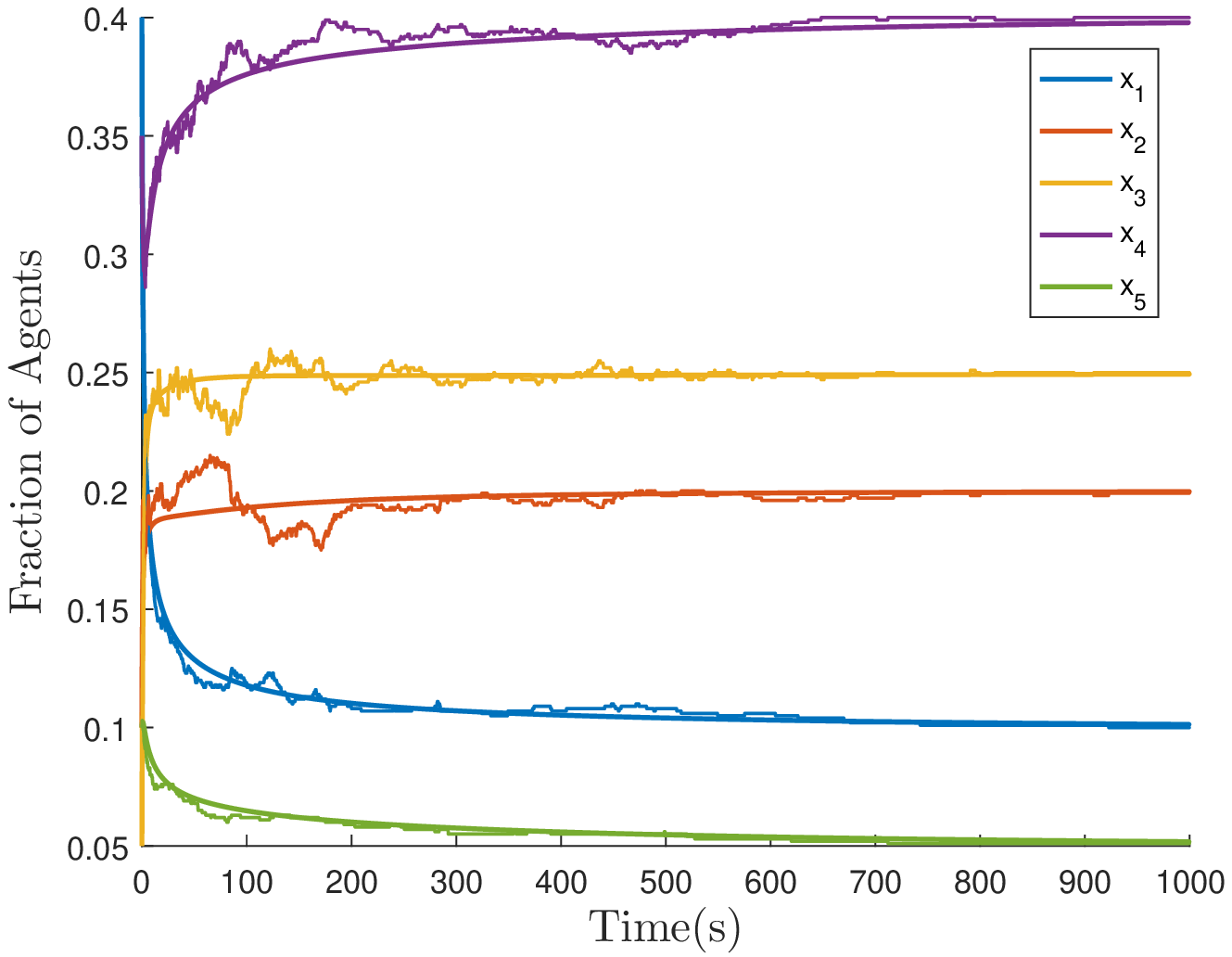}
	\caption[caption]{Trajectories of the mean-field model {\it (thick lines)} and the corresponding stochastic simulation {\it (thin lines)} for the {\it Case 1} closed-loop controller with $N=1000$ agents.}
	\label{fig:SoSControl_1000}
\end{figure}

\begin{figure}
	\centering
	\includegraphics[width= \linewidth]{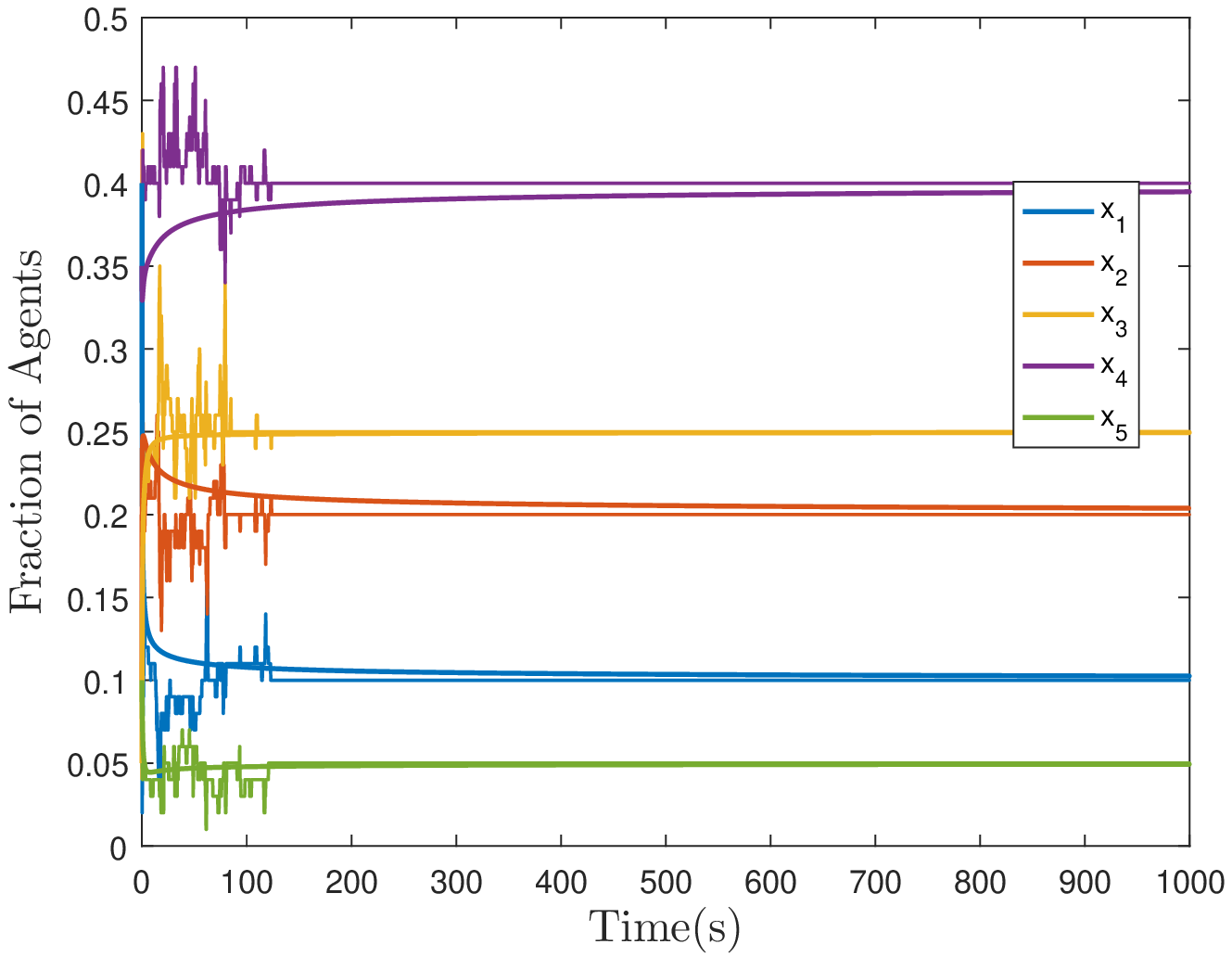}
	\caption[caption]{Trajectories of the mean-field model {\it (thick lines)} and the corresponding stochastic simulation {\it (thin lines)} for the {\it Case 2} closed-loop controller with $N=100$ agents.}
	\label{fig:ConstructControl_100}
\end{figure}

\begin{figure}
	\centering
	\includegraphics[width= \linewidth]{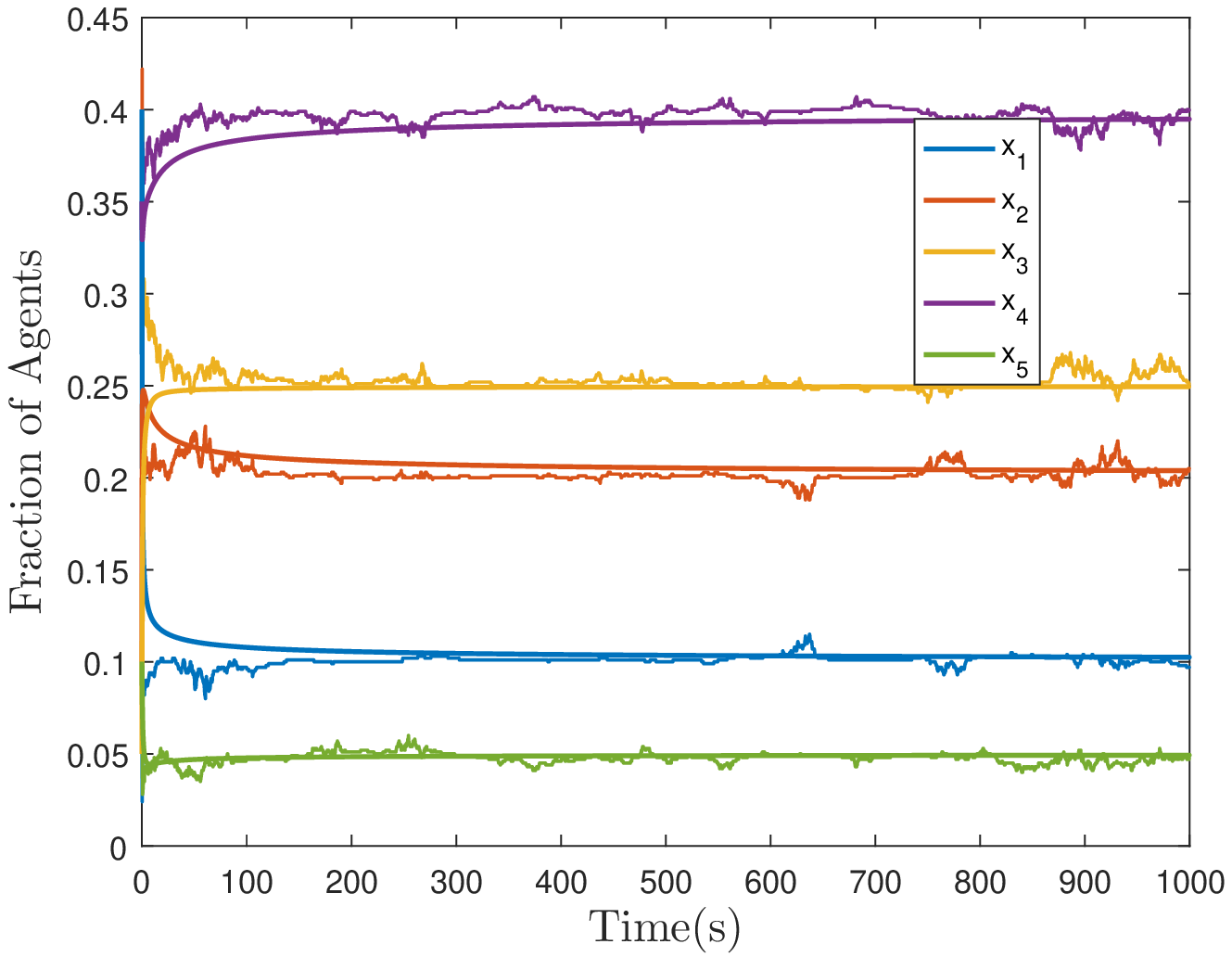}
	\caption[caption]{Trajectories of the mean-field model {\it (thick lines)} and the corresponding stochastic simulation {\it (thin lines)} for the {\it Case 2} closed-loop controller with $N=1000$ agents.}
	\label{fig:SOSControl_100}
\end{figure}

\begin{figure}
	\centering
	\includegraphics[width= \linewidth]{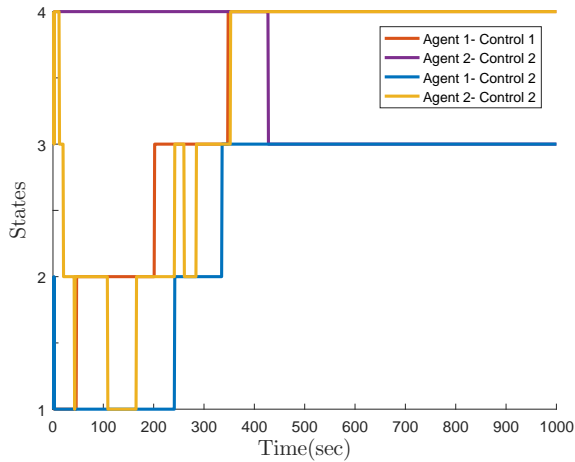}
	\caption[caption]{State (vertex number) of two agents over time, during stochastic simulations of the closed-loop system with the {\it Case 1} and {\it Case 2} controllers in the loop.}
	\label{fig:AgentPath}
\end{figure}

\section{CONCLUSION}
In this paper, we have presented a novel approach to mean-field feedback stabilization of a swarm of agents that switch stochastically among a set of states according to a continuous time Markov chain.  We proved that a desired state distribution with strongly connected support can be stabilized using time-invariant control inputs. We also showed asymptotic controllability of distributions that are not strictly positive, with target densities equal to zero for some states. Lastly, for bidirected, strongly connected graphs, we proved stabilizability of the closed-loop system by explicitly constructing decentralized, density-dependent control laws that equal zero at equilibrium. Furthermore, we presented and numerically validated a procedure for designing polynomial feedback control laws algorithmically using the SOSTOOLS MATLAB toolbox. In summary, by using nonlinear feedback control laws, we obtain guarantees on global boundedness of the controls and are able to prove global asymptotic stability of the desired distribution.

In future work, we plan to investigate exponential stability of the closed-loop system and design control laws that optimize the convergence rate to equilibrium. Another direction of future work is to characterize the effect of noise in estimates of the agent densities on the convergence properties of the proposed control laws.


\bibliographystyle{plain}
\bibliography{CDC2017}

\begin{thebibliography}{10}

\bibitem{acikmese2012markov}
Behcet Acikmese and David~S Bayard.
\newblock A {M}arkov chain approach to probabilistic swarm guidance.
\newblock In {\em American Control Conference (ACC)}, pages 6300--6307. IEEE,
  2012.

\bibitem{bandyopadhyay2013inhomogeneous}
Saptarshi Bandyopadhyay, Soon-Jo Chung, and Fred~Y Hadaegh.
\newblock Inhomogeneous markov chain approach to probabilistic swarm guidance
  algorithm.
\newblock In {\em 5th Int. Conf. Spacecraft Formation Flying Missions and
  Technologies,(Munich, Germany)}, 2013.

\bibitem{berman2009optimized}
Spring Berman, {\'A}d{\'a}m Hal{\'a}sz, M~Ani Hsieh, and Vijay Kumar.
\newblock Optimized stochastic policies for task allocation in swarms of
  robots.
\newblock {\em IEEE Transactions on Robotics}, 25(4):927--937, 2009.

\bibitem{boyd2004fastest}
Stephen Boyd, Persi Diaconis, and Lin Xiao.
\newblock Fastest mixing markov chain on a graph.
\newblock {\em SIAM review}, 46(4):667--689, 2004.

\bibitem{bullo2009distributed}
Francesco Bullo, Jorge Cort{\'e}s, and Sonia Martinez.
\newblock {\em Distributed control of robotic networks: a mathematical approach
  to motion coordination algorithms}.
\newblock Princeton University Press, 2009.

\bibitem{chapman2015advection}
Airlie Chapman.
\newblock Advection on graphs.
\newblock In {\em Semi-Autonomous Networks}, pages 3--16. Springer, 2015.

\bibitem{clarke1997asymptotic}
Francis~H Clarke, Yuri~S Ledyaev, Eduardo~D Sontag, and Andrei~I Subbotin.
\newblock Asymptotic controllability implies feedback stabilization.
\newblock {\em IEEE Transactions on Automatic Control}, 42(10):1394--1407,
  1997.

\bibitem{demir2015decentralized}
Nazl{\i} Demir, Utku Eren, and Beh{\c{c}}et A{\c{c}}{\i}kme{\c{s}}e.
\newblock Decentralized probabilistic density control of autonomous swarms with
  safety constraints.
\newblock {\em Autonomous Robots}, 39(4):537--554, 2015.

\bibitem{elamvaz2016lin}
Karthik Elamvazhuthi, Vaibhav Deshmukh, Kawski Matthias, and Spring Berman.
\newblock Mean field controllability and decentralized stabilization of markov
  chains, part i: Local controllability and rational feedbacks).

\bibitem{ethier2009markov}
Stewart~N Ethier and Thomas~G Kurtz.
\newblock {\em Markov processes: characterization and convergence}, volume 282.
\newblock John Wiley \& Sons, 2009.

\bibitem{gomes2013continuous}
Diogo~A Gomes, Joana Mohr, and Rafael~Rig{\~a}o Souza.
\newblock Continuous time finite state mean field games.
\newblock {\em Applied Mathematics \& Optimization}, 68(1):99--143, 2013.

\bibitem{halasz2007dynamic}
Ad{\'a}m Hal{\'a}sz, M~Ani Hsieh, Spring Berman, and Vijay Kumar.
\newblock Dynamic redistribution of a swarm of robots among multiple sites.
\newblock In {\em Intelligent Robots and Systems, 2007. IROS 2007. IEEE/RSJ
  International Conference on}, pages 2320--2325. IEEE, 2007.

\bibitem{kolokoltsov2010nonlinear}
Vassili~N Kolokoltsov.
\newblock {\em Nonlinear Markov processes and kinetic equations}, volume 182.
\newblock Cambridge University Press, 2010.

\bibitem{mather2011distributed}
T~William Mather and M~Ani Hsieh.
\newblock Distributed robot ensemble control for deployment to multiple sites.
\newblock {\em Proceedings of Robotics: Science and Systems VII}, 2011.

\bibitem{mesbahi2010graph}
Mehran Mesbahi and Magnus Egerstedt.
\newblock {\em Graph theoretic methods in multiagent networks}.
\newblock Princeton University Press, 2010.

\bibitem{norris1998markov}
James~R Norris.
\newblock {\em Markov chains}.
\newblock Number~2. Cambridge university press, 1998.

\bibitem{prajna2002introducing}
Stephen Prajna, Antonis Papachristodoulou, and Pablo~A Parrilo.
\newblock Introducing sostools: A general purpose sum of squares programming
  solver.
\newblock In {\em Decision and Control, 2002, Proceedings of the 41st IEEE
  Conference on}, volume~1, pages 741--746. IEEE, 2002.

\bibitem{puterman2014markov}
Martin~L Puterman.
\newblock {\em Markov decision processes: discrete stochastic dynamic
  programming}.
\newblock John Wiley \& Sons, 2014.

\bibitem{ren2008distributed}
Wei Ren and Randal~W Beard.
\newblock {\em Distributed consensus in multi-vehicle cooperative control}.
\newblock Springer, 2008.

\bibitem{schweighofer2005optimization}
Markus Schweighofer.
\newblock Optimization of polynomials on compact semialgebraic sets.
\newblock {\em SIAM Journal on Optimization}, 15(3):805--825, 2005.

\end{thebibliography}

\end{document}